\documentclass{article}
\usepackage[utf8]{inputenc}
\usepackage[letterpaper,top=2cm,bottom=2cm,left=3cm,right=3cm,marginparwidth=1.75cm]{geometry}
\usepackage{amsthm,amssymb,amsmath,color}
\usepackage{url}
\usepackage{hyperref}
\usepackage{enumerate}
\usepackage{natbib}

\usepackage{authblk}

\newcommand\numberthis{\addtocounter{equation}{1}\tag{\theequation}}
\newtheorem{theorem}{Theorem}
\newtheorem{claim}[theorem]{Claim}

\newtheorem{lemma}[theorem]{Lemma}
\newtheorem{corollary}[theorem]{Corollary}

\newtheorem{observation}[theorem]{Observation}

\newtheorem{example}[theorem]{Example}
\newtheorem{definition}[theorem]{Definition}

\newtheorem{assumption}{Assumption}

\renewcommand{\cite}{\citep}

\newcommand{\demand}{\bbx^{\mathcal{D}}}
\newcommand{\demandnb}{x^{\mathcal{D}}}

\newcommand{\RR}{\mathbb{R}}

\newcommand{\bba}{\mathbf{a}}
\newcommand{\bbg}{\mathbf{g}}
\newcommand{\bbp}{\mathbf{p}}
\newcommand{\bbq}{\mathbf{q}}
\newcommand{\bbx}{\mathbf{x}}

\newcommand{\bbb}{\mathbf{b}}
\newcommand{\bbv}{\mathbf{v}}

\newcommand{\calO}{\mathcal{O}}

\DeclareMathOperator*{\argmax}{arg\,max}

\title{Proportional Response Dynamics in Gross Substitutes Markets}

\author[1]{Yun Kuen Cheung}
\author[2]{Richard Cole} 
\author[3]{Yixin Tao}
\affil[1]{The Australian National University, Australia}
\affil[2]{New York University, USA}
\affil[3]{ITCS, Key Laboratory of Interdisciplinary Research of Computation and Economics\\ Shanghai University of Finance and Economics, China }

\begin{document}
\maketitle
\begin{abstract}
Proportional response is a well-established distributed algorithm which has been shown to converge to competitive equilibria in both Fisher and Arrow-Debreu markets, for various sub-families of homogeneous utilities, including linear and constant elasticity of substitution utilities. We propose a natural generalization of proportional response for gross substitutes utilities, and prove that it converges to competitive equilibria in Fisher markets. This is the first convergence result of a proportional response style dynamics in Fisher markets for utilities beyond the homogeneous utilities covered by the Eisenberg-Gale convex program.
We show an empirical convergence rate of $\calO(1/T)$ for the prices. Furthermore, we show that the allocations of  a lazy version of the generalized proportional response dynamics converge to competitive equilibria in Arrow-Debreu markets.
\end{abstract}

\section{Introduction}

Competitive (aka market) equilibrium is a fundamental concept in the study of markets, describing stable outcomes that can emerge from agents' trading. Originally proposed by L\'{e}on Walras more than 150 years ago~\cite{Walras1874}, its existence under mild conditions was formally proved by Nobel laureates Kenneth Arrow and G\'erald Debreu and Guggenheim fellow Lionel McKenzie in the 1950's~\cite{ArrowDebreu1954,McKenzie1954,McKenzie1959}. The early 2000's saw a surge of interest at the intersection of computer science and economics, leading researchers to adopt an algorithmic perspective on equilibrium. The relevance of algorithms to economics is neatly captured by Kamal Jain's remark: ``If your laptop cannot find it, neither can the market.''~\cite{Jain04} Nevertheless, real-world economic systems are often more akin to decentralized computing environments, where information is dispersed among agents who trade in a distributed manner, rather than relying on a central ``computer'' (the \emph{invisible hand}) with perfect information. Understanding how equilibrium can be achieved in such settings remains a challenge at the frontier of economic theory and algorithmic game theory.

Among various natural candidates of distributed market algorithms, proportional response (PR) has attracted significant attention, partly due to its simple implementation in networked markets~\cite{LLSB08}. The PR protocol is particularly intuitive in Fisher markets, where agents are divided into buyers with money and sellers with goods. It operates as an iterated process: in each round, buyers update their bids -- reflecting the amount of money they intend to spend on various goods -- based on their utility functions and the goods received in the previous round, while sellers allocate goods according to a Shapley-Shubik-style mechanism, distributing them in proportion to buyers’ bids~\cite{Branzei2021-Sigecom}. \citet{WZ2007} were the first to formally identify the power of PR in driving Arrow-Debreu markets toward competitive equilibrium. Since then, a growing body of research has extended PR's convergence guarantees to Arrow-Debreu markets, Fisher markets,
production economies and attention markets, demonstrating its robustness across a range of settings.~\cite{Zhang2011,BDX2011,cheung2018dynamics,BranzeiMehtaNisan2018NeurIPS,CheungHN19,gao2020first,BranzeiDevanurRabani2021,CheungLP2021,ZhuCX2023WWW,kolumbus2023asynchronous,li2024proportional,cheng2024tight}.

For PR in Fisher markets, all known convergence results have been established under the assumption that buyers have linear or constant elasticity of substitution (CES) utility functions, i.e., $u(\bbx) = \left(\sum_{j} a_{j} x_{j}^\rho\right)^{1/\rho}$. The competitive equilibria in such Fisher markets are characterized by the Eisenberg-Gale convex program~\cite{EisenbergGale,Eisenberg61}, which has provided valuable insights through connections to optimization algorithms. In fact, this convex program extends beyond CES utilities to capture equilibria for the broader class of homogeneous utility functions, i.e., $u(c  \bbx) = c\cdot u(\bbx)$ for any $c>0$ and good bundle $\bbx$. 
However, homogeneous utilities remain a relatively restrictive subset compared to the diverse preferences that economists have considered.
For instance, even the intuitive separable utilities of the form $u(\bbx) = \sum_j u_j(x_j)$ are generally not homogeneous.
This raises the open question: to what extent can PR be applied to markets with non-homogeneous utility functions?

The prior algorithmic success in competitive equilibrium computation suggests that gross substitute (GS) utilities are a promising target for addressing the open question. GS utilities emerged as a natural sufficient condition for the convergence of \emph{t\^atonnement}~\cite{ArrowHurwicz1958,ABH1959},
a natural market dynamics introduced by~\citet{Walras1874} alongside the concept of competitive equilibrium.
While the GS assumption enables an intuitive proof of t\^atonnement convergence,  it is not clear what the PR update rule for GS utilities should be, and whether it is sufficient to guarantee the convergence of PR.

\paragraph{Our contributions.} We propose a new PR rule that extends to any differentiable utilities in Fisher markets.
This generalizes the PR rule for substitute CES utilities.
In each round, a buyer with utility $u$ allocates monetary bids across each good $j$ in proportion to $x_j \nabla_j u(\bbx)$, where $\bbx$ is the good bundle she received in the previous round. Precisely, if the buyer's budget is $e$, her bid for good $j$ is
\begin{align*}
b_j ~=~ e \frac{x_j \nabla_j u(\bbx)}{\sum_{j'} x_{j'} \nabla_{j'} u(\bbx)}~. \numberthis \label{eqn:formula-intro}
\end{align*}
Based on the bids from all buyers, each good is allocated in proportion to the bids placed on that good. Here are a few remarks to our new PR rule:
\begin{enumerate}
    \item The PR rule \eqref{eqn:formula-intro} operates without a step-size parameter, distinguishing it from other market iterative algorithms like t\^atonnment and ascending-price auctions. 
    \item When $u$ is a CES utility function of the form $u(\bbx) = \left(\sum_{j'} a_{j'} x_{j'}^\rho\right)^{1/\rho}$ where $0<\rho\le 1$, the PR update \eqref{eqn:formula-intro} simplifies to $b_j = e \frac{a_j (x_j)^\rho}{\sum_{j'} a_j x_{j'}^\rho}$, which aligns with the PR rule analyzed in \cite{Zhang2011} and \cite{cheung2018dynamics}.
    \item When $u$ is a homogeneous utility function, i.e., $u(c\bbx) = c\cdot u(\bbx)$ for any $c>0$ and good bundle $\bbx$, by Euler's homogeneous function theorem $u(\bbx) = \sum_{j'} x_{j'} \nabla_{j'} u(\bbx)$, so the PR update \eqref{eqn:formula-intro} simplifies to $b_j = e \frac{x_j \nabla_j u(\bbx)}{u(\bbx)}$.
    \item For general utility function $u$, the vector $\bbq = \left\{e \frac{\nabla_j u(\bbx)}{\sum_{j'} x_{j'} \nabla_{j'} u(\bbx)}\right\}_{j}$ can be interpreted as a price vector at which $\bbx$ is the optimal demand bundle. Precisely, 
    $\bbx$ is the optimal solution to $\arg \max_{\bbx'} u(\bbx')$ subject to the budget constraint $ \bbq^\top \bbx' \leq e$.
    Thus, the PR update rule can be interpreted as follows: given the allocation from the previous round, the agent assumes the current price to be the one at which this allocation is the best response. In the next round, the agent adjusts spending accordingly to optimize demand under this assumed price. \label{enu:3}
\end{enumerate} 

We show that if all buyers have GS utilities satisfying a few standard conditions, the new PR guarantees an empirical convergence rate of $\calO(1/T)$ toward competitive equilibrium prices; that is, the distance between the average of the prices over the first $T$ iterations and the equilibrium prices is bounded by $\calO(1/T)$. Our analysis also shows that the new PR converges point-wise to competitive equilibrium, but it does not provide a convergence rate.

To extend our analysis to Arrow-Debreu markets, we consider a \emph{lazy} version of the new PR update.
The idea of lazy PR was first proposed by \citet{BranzeiDevanurRabani2021}: each agent maintains a bank account where she accumulates all revenues from selling her endowment of goods. Instead of spending her entire balance in each round, she allocates only a fixed fraction (strictly between 0 and 1) of her savings.
We show the allocations of this lazy PR update converge to the allocation of a competitive equilibrium.

\vspace{0.1in}
Our proof strategy involves transforming the update rule into a price-demand system for each agent, as outlined in Remark (4).
While the price vectors may differ across agents, the corresponding demand (best response) system provides a more intuitive framework for analysis and applying the GS property. A central component of our approach is Lemma~\ref{lem:main-tech}, which establishes a new inequality in the context of gross substitutes markets characterizing the equilibrium allocation and price. Specifically, consider a scenario where each agent $i$ has a personalized price vector $\bbq_i$ and selects an optimal demand bundle $\bbx_i$ based on this price vector. If the total demand bundle remains feasible, i.e., $\sum_i x_{ij} = 1$ \footnote{WLOG, we assume each good has a unit supply.} for all goods $j$, then the following inequality holds: $\sum_{ij} x_{ij}^* p_j^* \log p_j^* \leq \sum_{ij} x_{ij}^* p_j^* \log q_{ij}$, where $x_{ij}^*$ and $p_j^*$ represent the market equilibrium allocation and price respectively. Notably, the equilibrium price vector $\bbp^*$ can be interpreted as a common price vector for all agents, with the corresponding equilibrium allocation $\bbx^*$ remaining feasible. This inequality can be viewed as defining a potential function, where the equilibrium allocation and price correspond to the minimum of this function.

\paragraph{Further Related Work.}
The study of market dynamics has long been a fundamental topic in both economics and algorithmic game theory. 

One of the foundational dynamics in this area is the t\^atonnement process, introduced by~\citet{Walras1874}, which describes a natural market adjustment driven by supply and demand. In this process, agents always obtain their preferred bundles of goods while prices adjust dynamically—rising for overdemanded goods and falling for underdemanded ones. The development of the t\^atonnement procedure progressed with the advancement of general equilibrium theory. \citet{samuelson1983foundations} modeled t\^atonnement as a continuous dynamical system, laying the foundation for subsequent theoretical studies. In seminal works by \citet{ArrowHurwicz1958} and \citet{ABH1959}, the gross substitutes property was introduced as a crucial condition to ensure the convergence of the t\^atonnement process. In the early 2000s, increasing interest at the intersection of computer science and economics led researchers to adopt an algorithmic perspective on market equilibrium. This shift prompted the exploration of discrete variants of the t\^atonnement algorithms. See e.g., \cite{CMV2005, ARY14, goktas2023tatonnement, CF2008, CCR2012, CCD2013, cole2019balancing, goktas2021consumer, dvijotham2022convergence, gao2020first, fleischer2008fast, CheungCole2018async}.

Many researchers have explored an even simpler market dynamic: auction algorithms. While tâtonnement adjusts prices bidirectionally in response to supply and demand, auction algorithms follow an ascending price trajectory. The first auction algorithm was introduced by \citet{GargKap2006} for linear utility functions and was later extended to broader settings \cite{garg2006price, garg2007market, garg2004auction}, including separable concave gross substitutes utility functions and uniformly separable gross substitutes utility functions. For general weak gross substitutes (WGS) utilities, ascending-price algorithms have been studied by \citet{bei2019ascending} and \citet{garg2023auction}.

Most of these results fall within the domain of gross substitutes utility functions \cite{fisher1972gross}. However, outside this domain, certain classes of non-gross substitutes utility functions have been shown to pose significant computational challenges for finding market equilibria in either Fisher markets or Arrow-Debreu markets \cite{chen2009spending, CDDT2009, CPY2017, CSVY2006, garg2017settling, VazYan10, bei2019earning}. Despite these difficulties, substantial research effort has been dedicated to developing algorithms for computing market equilibria under various settings \cite{brainard2005compute, codenotti2005polynomial, codenotti2004computation, DPSV08, duan2015combinatorial, garg2023strongly, jain2006equilibria, Ye04}.

\paragraph{Roadmap} The rest of the paper is structured as follows. Section~\ref{sec:setting-fisher markets} introduces the model and definitions for Fisher markets. Section~\ref{sec:prd-fisher} defines the new PR rule and demonstrates its convergence to equilibrium. Section~\ref{sec:exchange-market} extends the Fisher markets result to exchange markets and shows that the lazy version of the new PR rule also converges to equilibrium.

\paragraph{Notations}
In this paper, we use bold letters like $\bbp,\bbx,\bbb$ to denote vectors. When we write $\bbv > 0$ for a vector $\bbv$, we mean every component of $\bbv$ is larger than $0$.
\section{Setting of Fisher markets} \label{sec:setting-fisher markets}

We consider a Fisher market with $n$ buyers and $m$ goods, where each good has a unit supply. Each buyer has a budget $e_i$ and  a  utility function $u_i(\cdot) : \RR_{\geq 0}^{m} \rightarrow \RR$.  We make the following assumptions regarding the utility functions:
\begin{assumption}[Utility function Properties] \label{aspt:u-std}
    The utility function $u_i(\cdot)$ satisfies: 
    \begin{itemize}
        \item $u_i(\cdot)$ is strictly concave;
        \item $u_i(\cdot)$ is strictly increasing. 
    \end{itemize}
\end{assumption}

\paragraph{Buyer Demand}
Given a price vector of goods $\bbp > 0$, and a budget $e_i$, the demand of buyer $i$ is defined as:
\[ \demand_i(\bbp, e_i) \triangleq \argmax_{\bbx_i} u_i(\bbx_i) \text{ s.t. } \bbp^\top \bbx_i \leq e_i \text{ and } \bbx_i \geq 0. \]

From Assumption~\ref{aspt:u-std}, we derive the following observation.

\begin{observation}[Full Budget Spending]\label{obs:full-spending}
    For any $\bbp > 0$, buyer $i$ will spend the entire budget, i.e., $\bbp^\top \demand_i(\bbp, e_i) = e_i$.
\end{observation}

\paragraph{Gross Substitutes and Normal Goods} We consider utility functions that satisfy the Gross Substitutes (GS) property and assume that all goods are normal. These assumptions are formalized below:
\begin{assumption} \label{def:AGS}
The utility function $u_i(\cdot)$ of each buyer $i$ in the market satisfies the following property, i.e.,
\begin{itemize}
    \item  \textbf{Gross Substitutes (GS):}  For any price vectors $\bbp$ and $\bbp'$ such that $\bbp \leq \bbp'$,  and for any good $j$ with $p_j = p'_j$:
$$\demandnb_{ij}(\bbp, e_i) \leq \demandnb_{ij}(\bbp', e_i).$$
    \item \textbf{Normal Goods:} For any budgets $e_i$ and $e_i'$ such that $e_i \leq e_i'$: $$\demandnb_{ij}(\bbp, e_i) \leq \demandnb_{ij}(\bbp, e'_i). $$
\end{itemize}
\end{assumption}

\paragraph{Corresponding Price} Since the utility function is concave, given any allocation $\bbx_i > 0$, there exists a price vector $\bbp$ such that $\bbx_i$ is the demand of buyer $i$. This is formalized in the following lemma.
\begin{lemma}
    For any allocation $\bbx_i > 0$, there exists a price $\bbp$ such that $\bbx_i = \demand_i(\bbp, e_i)$.
\end{lemma}
\begin{proof}
    Consider any subgradient $\bbg \in \nabla u_i(x_i)$, and construct the price vector $\bbp$ by setting $p_j = e_i g_j / (\sum_j g_j x_{ij})$. We aim to show $\bbx_i$ is a maximizer of $u_i(\bbx_i)$ subject to $\bbp^\top \bbx_i \leq e_i$. This follows from the Karush-Kuhn-Tucker(KKT) conditions, where the dual variable associated with the constraint $\bbp^\top \bbx_i \leq e_i$ is given by $\lambda = (\sum_j g_j x_{ij}) / e_i$. 
\end{proof}

Additionally, the following lemma shows that this price vector is unique.
\begin{lemma} \label{lem:tech-1}
Given an allocation $\bbx_i$, and two price vectors $\bbp > 0$ and $\bbp' > 0$ such that $\bbx_i = \demand_i(\bbp, e_i) = \demand_i(\bbp', e_i)$,
if $x_{ij} > 0$ for good $j$, then $p_j = p'_j$. 
\end{lemma}
\begin{proof}
    Let $\mathcal{Z}^+ = \{j ~|~ p_j > p'_j\}$ and $\mathcal{Z}^- = \{j ~|~ p_j <  p'_j\}$. Adjusting prices $p_j$ for $j \in \mathcal{Z}^-$ to $p'_j$ will not decrease the total spending on goods in $\mathcal{Z}^+$ due to the GS property. Similarly, adjusting the prices $p_j$ for $j \in \mathcal{Z}^+$ to $p'_j$ will also not decrease the total spending of goods in $\mathcal{Z}^+$. Consequently, the total spending of goods in $\mathcal{Z}^+$  does not decrease when moving from price vector $\bbp$ to $\bbp'$. Given that $p_j > p'_j$ for $j \in \mathcal{Z}^+$, this implies $x_{ij} = 0$ for $j \in \mathcal{Z}^+$. Similarly, one can show $x_{ij} = 0$ for $j \in \mathcal{Z}^-$.
\end{proof}

Therefore, we can define the corresponding price $\bbq_i(\bbx_i)$, which is the price vector such that $\bbx_i$ is the demand of buyer $i$.

\begin{definition}[Corresponding Price]
    For any $\bbx_i > 0$, the corresponding price vector $\bbq_i(\bbx_i)$ is defined as:
    \begin{align*}
        q_{ij}(\bbx_i) = e_i \frac{ \nabla_j u_i(\bbx_i)}{\sum_{j'} x_{ij'} \nabla_{j'} u_i(\bbx_i)}. \numberthis \label{eqn:corres-price}
    \end{align*}
\end{definition}
\begin{observation} \label{obs:personal-demand-price}
For any $\bbx_i > 0$, $\bbx_i = \demand_i(\bbq_i(\bbx_i), e_i)$.
\end{observation}

We note that corresponding prices may not be well-defined when $\bbx_i$ is not strictly positive. For instance, if $u_i$ is a CES utility function of the form $u_i(\bbx_i) = \left(\sum_j a_{ij} (x_{ij})^\rho\right)^{1/\rho}$ where $0<\rho<1$, then $\nabla_j u_i(\bbx_i)=+\infty$ when $x_{ij}=0$ and $a_{ij} > 0$. Suppose a strictly positive sequence $\{\bbx_i^t\}$ converges to $\bbx_i$, where $x_{ij}=0$ for some goods $j$. Although each $\bbq_i(\bbx_i^t)$ is well-defined,
the sequence $\{\bbq_i(\bbx_i^t)\}$ may not converge, or it may converge to different possible values depending on the directions of convergence of the sequence  $\{\bbx_i^t\}$. The possibility of such situation introduces further challenge for our analysis; we will highlight it as we proceed.

\paragraph{Fisher market equilibrium}
In a Fisher market, each good $j$ has a price $p_j$ and each buyer $i$ receives a bundle $\bbx_i$, where $x_{ij}$ denotes the allocation of good $j$ to buyer $i$. The concept of Fisher market equilibrium characterizes the equilibrium allocation and prices in the market. WLOG, we assume that each good has one unit of supply.
\begin{definition}[Fisher market equilibrium]
    An allocation-price pair $(\bbx^*, \bbp^*)$ is a Fisher market equilibrium if
    \begin{enumerate}
        \item each buyer $i$ receives the optimal bundle given the price, $\bbx^*_i \in \demand_i(\bbp^*, e_i)$;
        \item no good is oversold, $\sum_{i} x^*_{ij} \leq 1$ for all $j$;
        \item every good with a positive price is fully allocated, $\sum_{i} x^*_{ij} < 1$ implies $p^*_j = 0$.
    \end{enumerate}
\end{definition}
Since utility functions are strictly increasing, the following observation holds.
\begin{observation} \label{obs:pos-price}
    At equilibrium, $\bbp^* > 0$ and for each good $j$, $\sum_i x^*_{ij} = 1$. 
\end{observation}
\section{Proportional Response Dynamics in Fisher markets} \label{sec:prd-fisher}
Proportional Response Dynamics is an iterative procedure. At each iteration $t$, the dynamics are characterized by the following steps: 
\begin{enumerate}
    \item Budget Spending: Each buyer $i$ allocates their budget $e_i$ among the goods. Let $b^t_{ij}$ denotes the spending of buyer $i$ on good $j$ at round $t$. It follows that $\sum_j b_{ij}^t = e_i$ for each buyer $i$.
    \item Goods Allocation: Based on the spending $\{b^t_{ij}\}_{\{i, j\}}$, the amount of good $j$ allocated buyer $i$ is proportional to their spending on good $j$: $x_{ij}^t = b^t_{ij} / (\sum_i b_{ij}^t)$.
    \item Spending Update: For the next round, the spending $b_{ij}^{t+1}$ is updated to be proportional to $x_{ij}^t \nabla_j u_i(\bbx_i^t)$.
\end{enumerate}
The process begins with an initial spending allocation $b_{ij}^0 > 0$ for all buyer $i$ and good $j$.

Mathematically, the entire procedure can be expressed as:
\begin{align*}
    &b^{t+1}_{ij} = e_i \frac{x_{ij}^t \nabla_j u_i(\bbx_i^t)}{\sum_{j'} x_{ij'}^t \nabla_{j'} u_i(\bbx_i^t)}\hspace*{0.1in}\text{ for all agent $i$ and good $j$,}\\
    &~~~~~~\hspace{0.4in} \text{where $x_{ij}^t = \frac{b_{ij}^t}{p_j^t}$ and $p_j^t = \sum_i b_{ij}^t$.} \numberthis \label{eq:pr}
\end{align*}
The update rule can be rewritten as follows:
    $$b_{ij}^{t+1} = x_{ij}^t q_{ij}(\bbx_i^t) \hspace*{0.1in}\text{ for all agent $i$ and good $j$},$$ 
    where $q_{ij}(\bbx_i^t)$ is the corresponding price of allocation $\bbx_i^t$ (see \eqref{eqn:corres-price}). This implies that when buyer $i$ receives bundle $\bbx_i^t$ in round $t$, in next round, buyer $i$ will make the optimal spending assuming the price is the corresponding price of $\bbx_i^t$.

\begin{example}
 When $u_i$ is a homogeneous utility function, i.e., $u_i(c\bbx_i) = c\cdot u_i(\bbx_i)$ for any $c > 0$ and bundle $\bbx_i$, by Euler's homogeneous function theorem, $u_i(\bbx_i) = \sum_{j'} x_{ij'} \nabla_{j'} u_i(\bbx_i)$, the update rule simplifies to $b_{ij}^{t+1} = e_i \frac{x_{ij}^t \nabla_{j} u_i(\bbx_i^t)}{u_i(\bbx_i^t)}$. Moreover, for CES utility functions, $u_i(\bbx_i) = (\sum_j a_{ij} x_{ij}^{\rho_i})^{1 / \rho_i}$, the update rule can further be simplified to: $b^{t+1}_{ij} = e_i \frac{a_{ij} (x_{ij}^t)^{\rho_i}}{\sum_{j'} a_{ij'} (x_{ij'}^t)^{\rho_i}}.$
\end{example}

\subsection{Convergence Results}
Our main theorems demonstrate that the prices and allocations in Proportional Response Dynamics converge to a market equilibrium.
\begin{theorem} \label{thm:main-1}
The sequence of price vectors $\bbp^t$ converges to the equilibrium price $\bbp^*$. Additionally, the average price converges at a rate of $\calO(1/T)$.
\end{theorem}
\begin{theorem}\label{thm:main-2}
The sequence of allocations $\bbx^t$ converges to the equilibrium allocation $\bbx^*$.
\end{theorem}
The proof of Theorem~\ref{thm:main-1} is presented in Section~\ref{sec:proof-convergence-1}, and the proof of Theorem~\ref{thm:main-2} is presented in Section~\ref{sec:proof-convergence}.

\subsection{Proof of Theorem~\ref{thm:main-1}} \label{sec:proof-convergence-1}

To prove Theorem~\ref{thm:main-1}, we start with a simple observation.
\begin{observation} \label{obs:b-x-pos}
At any round $t$, both the spending $b_{ij}^t$ and the allocation $x_{ij}^t$ are strictly positive for any buyer $i$ and good $j$. 
\end{observation}

The following key technical lemma focuses on the concept of personal pricing, where each buyer $i$ is assigned a distinct price vector, $\bbq_i$. We consider all possible personal pricing vectors such that the allocation remains feasible, i.e., $\sum_i \demandnb_{ij}(\bbq_i, e_i) = 1$ for every good $j$. The equilibrium price vector is a specific instance of such a personal pricing choice. The following lemma, proven in Section~\ref{sec:main-tech}, demonstrates that the equilibrium price vector $\bbp^*$ minimizes the expression $\min_{\bbq} \sum_{ij} \demandnb_{ij}(\bbp^*, e_i) p^*_j \log q_{ij}$ among all feasible personal pricing choices.
\begin{lemma} \label{lem:main-tech}
    Let $\bbq_i > 0$ denote the personal price vector for buyer $i$. If $\sum_i \demandnb_{ij}(\bbq_i, e_i) = 1$ for all goods $j$, then 
    \begin{align*}
        \sum_{ij} \demandnb_{ij}(\bbp^*, e_i) p^*_j \log p^*_j \leq  \sum_{ij} \demandnb_{ij}(\bbp^*, e_i) p^*_j \log q_{ij}.
    \end{align*}
    Equality holds only if $\demand_i(\bbq_i, e_i) = \demand_i(\bbp^*, e_i)$ for all buyers.
\end{lemma}

\begin{proof}[Proof of Theorem~\ref{thm:main-1}]
    To establish the convergence of prices, we analyze the KL divergence between the equilibrium spending $b_{ij}^*$ and the spending at round $t+1$:
    \begin{align*}
        \sum_{ij} b_{ij}^* \log \frac{b_{ij}^*}{b_{ij}^{t+1}} &= \sum_{ij} b_{ij}^* \log \frac{b_{ij}^*}{e_i \frac{x_{ij}^t \nabla_j u_i(\bbx_i^t)}{\sum_{j'} x_{ij'}^t \nabla_{j'} u_i(\bbx_i^t)}} ~~~~~~~~~\hspace{1in} \text{(by \eqref{eq:pr})}\\
        &= \sum_{ij} b_{ij}^* \log \frac{b_{ij}^*}{b_{ij}^t} - \sum_{ij} b_{ij}^* \log \frac{p_j^*}{p_j^t} + \sum_{ij} b_{ij}^*  \log \frac{p_j^*}{e_i \frac{ \nabla_j u_i(\bbx_i^t)}{\sum_{j'} x_{ij'}^t \nabla_{j'} u_i(\bbx_i^t)}} \\
        &= \sum_{ij} b_{ij}^* \log \frac{b_{ij}^*}{b_{ij}^t} - \sum_{j} p_{j}^* \log \frac{p_j^*}{p_j^t} + \sum_{ij} b_{ij}^*  \log \frac{p_j^*}{e_i \frac{ \nabla_j u_i(\bbx_i^t)}{\sum_{j'} x_{ij'}^t \nabla_{j'} u_i(\bbx_i^t)}}. \numberthis \label{eq:3}
    \end{align*}
    Using Lemma~\ref{lem:main-tech}, we will show that $\sum_{ij} b_{ij}^*  \log \frac{p_j^*}{e_i \frac{ \nabla_j u_i(\bbx_i^t)}{\sum_{j'} x_{ij'}^t \nabla_{j'} u_i(\bbx_i^t)}}$ is non-positive, for this implies that the KL divergence for the spending decreases monotonically.

    In order to apply Lemma~\ref{lem:main-tech}, we use the corresponding price vector $\bbq_i(\cdot)$ for buyer $i$:
    \begin{align*}
    q_{ij}(\bbx_i^t) = e_i \frac{ \nabla_j u_i(\bbx_i^t)}{\sum_{j'} x_{ij'}^t \nabla_{j'} u_i(\bbx_i^t)}. 
    \end{align*}
    
    By Observation~\ref{obs:personal-demand-price}, and given that $\bbx_i^t > 0$, we have $\bbx_i^t = \demand_{i}(\bbq_i(\bbx_i^t), e_i)$, and this implies $\sum_i \demandnb_{ij}(\bbq_i(\bbx_i^t), e_i) = \sum_i x_{ij}^t = 1$ for all goods $j$. 
    Therefore, by Lemma~\ref{lem:main-tech}, $$\sum_{ij} b_{ij}^*  \log \frac{p_j^*}{e_i \frac{ \nabla_j u_i(\bbx_i^t)}{\sum_{j'} x_{ij'}^t \nabla_{j'} u_i(\bbx_i^t)}} = \sum_{ij} x_{ij}^* p_j^* \log \frac{p_j^*}{q_{ij}(\bbx_i^t)} \leq 0.$$ 
    
    Thus, the KL divergence decreases monotonically,
    \begin{align*}
        \sum_{ij} b_{ij}^* \log \frac{b_{ij}^*}{b_{ij}^{t+1}} \leq \sum_{ij} b_{ij}^* \log \frac{b_{ij}^*}{b_{ij}^t} - \sum_{j} p_{j}^* \log \frac{p_j^*}{p_j^t}. \numberthis \label{eq:kl-decrease}
    \end{align*}
    This implies the sequence $\bbp^t$ converges to $\bbp^*$. Furthermore, by the convexity of the KL divergence:
    \begin{align*}
        \sum_{j} p_j^* \log \frac{p_j^*}{\frac{1}{T}\sum_t p_j^t} \leq \frac{1}{T} \sum_{t = 1}^T\sum_{ij} p_j^* \log \frac{p_j^*}{ p_j^t} \overset{(a)}{\leq} \frac{1}{T}  \sum_{ij}b_{ij}^* \log \frac{b_{ij}^*}{b_{ij}^0},
    \end{align*}
    (a) follows by summing \eqref{eq:kl-decrease} over $t$.
    This establishes the convergence rate of $O(1/T)$ for the average price vector.

\end{proof}

\subsection{Proof of Lemma~\ref{lem:main-tech}} \label{sec:main-tech}
In this section, we provide the proof of Lemma~\ref{lem:main-tech}. The proof begins by decomposing Lemma~\ref{lem:main-tech} into individual buyers.
\begin{lemma} \label{lem:main-technique}
    For any GS utility $u_i(\cdot)$, and for any price vectors  $\bbp > 0$ and $\bbq > 0$, the following inequality holds:
    \[ \sum_j p_j \demandnb_{ij}(\bbp, e_i) \log \frac{p_j}{q_j}  \leq \sum_j p_j\left[\demandnb_{ij}(\bbq, e_i) - \demandnb_{ij}(\bbp, e_i)\right]. \]
    Equality occurs only if $\demand_{i}(\bbq, e_i) = \demand_{i}(\bbp, e_i)$.
\end{lemma}

Using Lemma~\ref{lem:main-technique}, we now proceed to prove Lemma~\ref{lem:main-tech}.
\begin{proof}[Proof of Lemma~\ref{lem:main-tech}]
By Lemma~\ref{lem:main-technique},
 \begin{align*}
        \sum_{ij} \demandnb_{ij}(\bbp^*, e_i) p^*_j \log \frac{p^*_j}{q_{ij}} &\leq  \sum_{ij} p^*_j\left[\demandnb_{ij}(\bbq_i, e_i) - \demandnb_{ij}(\bbp^*, e_i)\right]  = \sum_{j} p^*_j(1 - 1)  = 0. 
\end{align*}
\end{proof}

\begin{proof}[Proof of Lemma~\ref{lem:main-technique}]
The proof of Lemma~\ref{lem:main-technique} proceeds by analyzing a price adjustment procedure that begins with the initial price vector $\bbq$ and ends with the final price vector $\bbp$. This procedure is divided into several steps. In each step, the current price vector $\bbp'$ is updated to a new price vector $\bbp''$ by applying one of the following two operations:
\begin{itemize}
    \item Operation $1$: Let $\mathcal{S}$ be the set of the goods with the \emph{maximal} ratio $\frac{p_j}{p'_j} > 1$, if any. \emph{Increase} the $p'_j$ to $p''_{j}$ proportionally for all $j \in \mathcal{S}$  until the ratio $\frac{p_j}{p''_{j}}$ matches the second maximal ratio or $1$.
    \item Operation $2$: Let $\mathcal{S}$ be the set of the goods with the \emph{minimal} ratio $\frac{p_j}{p'_j} < 1$, if any. \emph{Decrease} the $p'_j$ to $p''_{j}$ proportionally  for all $j \in \mathcal{S}$ until the ratio $\frac{p_j}{p''_{j}}$ matches the second minimal ratio or $1$.
\end{itemize}

For both operations, the following inequality is established: 
\begin{lemma} \label{lem:min-side}
        $\sum_j p_j \demandnb_{ij}(\bbp, e_i) \log \frac{p''_j}{p'_j} \leq \sum_{j} p_j(\demandnb_{ij}(\bbp', e_i) - \demandnb_{ij}(\bbp'', e_i)) $. Equality occurs only if $\demand_{i}(\bbp', e_i) = \demand_{i}(\bbp'', e_i)$.
\end{lemma}

Once Lemma~\ref{lem:min-side} is proven, Lemma~\ref{lem:main-technique} follows by summing the terms over all steps and performing a telescoping sum.

\paragraph{Proof of Lemma~\ref{lem:min-side} for Operation $1$}
 
    We start by introducing some additional notation. Let $\mathcal{S} \in [m]$ denote the set of goods with the largest ratio between $\bbp$ and $\bbp'$, let $r$ denote this ratio, and let $r'$ denote the largest ratio for goods not in $\mathcal{S}$ (or $1$ if all the goods are in $\mathcal{S}$):
    \[ \mathcal{S} = \argmax_j p_j / p'_j,~~~ r = \max_j p_j / p'_j, ~~~\text{and}~~~ r' = \max \left\{ \max_{j \notin \mathcal{S}} p_j / p'_j ~,~ 1 \right\}.\]
    The price $\bbp''$ is constructed as follows:
    \begin{align*}
        p''_j = \begin{cases}
        p_j / r' & \text{if } j \in \mathcal{S}\\
        p'_j & \text{o.w. } .
    \end{cases}
    \end{align*}
    We begin with the following claim.
    \begin{claim}\label{clm:spending}
        $\sum_{j \in \mathcal{S}} p_j \demandnb_{ij}(\bbp, e_i) \leq r' \sum_{j \in \mathcal{S}}  p''_j \demandnb_{ij}(\bbp'', e_i)$.
    \end{claim}
    \begin{proof}
        Consider the term $\demandnb_{ij}(\bbp / r', e_i / r')$. Since both prices and budget are scaled by the same factor, we have $\demandnb_{ij}(\bbp / r', e_i / r') = \demandnb_{ij}(\bbp, e_i)$. Observe that $\bbp / r'$ and $\bbp''$ are identical for those goods in $\mathcal{S}$.  Now, increase the prices of goods $j \notin \mathcal{S}$ from $\bbp / r'$ to $\bbp''$ and increase the budget from $e_i / r'$ to $e_i$.
        By Assumption~\ref{def:AGS}, these adjustments increase the demand for goods in $S$. Thus, $\demandnb_{ij}(\bbp, e_i) \leq \demandnb_{ij}(\bbp'', e_i)$ for $j \in \mathcal{S}$. The claim follows as $p_j = r' p''_j$ by our construction of $\bbp''$ for $j \in \mathcal{S}$.
    \end{proof}
    Now, we aim to establish the following inequality: 
    \begin{align*}
        r' \sum_{j\in \mathcal{S}} p''_j \demandnb_{ij}(\bbp'', e_i) \log (r / r') \leq  \sum_{j} p_j\left[\demandnb_{ij}(\bbp', e_i) - \demandnb_{ij}(\bbp'', e_i)\right].\numberthis \label{eqn::1}
    \end{align*}
    This inequality directly implies Lemma~\ref{lem:min-side} when combined with Claim~\ref{clm:spending}. (Note that the terms in the LHS sum in Lemma~\ref{lem:min-side} are $0$ for $j \notin \mathcal{S}$) 

    We begin by deriving a lower bound for the RHS of the inequality:
    \begin{align*}
        &\sum_{j} p_j\left[\demandnb_{ij}(\bbp', e_i) - \demandnb_{ij}(\bbp'', e_i)\right]\\
        &~~~~~~=\sum_{j\in \mathcal{S}} p_j\left[\demandnb_{ij}(\bbp', e_i) - \demandnb_{ij}(\bbp'', e_i)\right] + \sum_{j \notin \mathcal{S}} p_j\left[\demandnb_{ij}(\bbp', e_i) - \demandnb_{ij}(\bbp'', e_i)\right] \\
        &~~~~~~=\sum_{j \in \mathcal{S}} r' p''_j \left[\demandnb_{ij}(\bbp', e_i) - \demandnb_{ij}(\bbp'', e_i)\right] + \sum_{j \notin \mathcal{S}} p_j \left[\demandnb_{ij}(\bbp', e_i) - \demandnb_{ij}(\bbp'', e_i)\right]  \\
        &~~~~~~\geq \sum_{j \in \mathcal{S}} r' p''_j \left[\demandnb_{ij}(\bbp', e_i) - \demandnb_{ij}(\bbp'', e_i)\right] + \sum_{j \notin \mathcal{S}} r' p''_j\left[\demandnb_{ij}(\bbp', e_i) - \demandnb_{ij}(\bbp'', e_i)\right] \\
        &~~~~~~ = \sum_{j} r' p''_j\left[\demandnb_{ij}(\bbp', e_i) - \demandnb_{ij}(\bbp'', e_i)\right].
    \end{align*}
    The second equality follows from our construction of $\bbp''$, and the inequality holds because $p_j \leq r' p''_j$ and $\demandnb_{ij}(\bbp', e_i) \leq \demandnb_{ij}(\bbp'', e_i)$ for $j \notin \mathcal{S}$, as implied by the GS property. 
    
    Using Observation~\ref{obs:full-spending}, the above expression can be further simplified:
    \begin{align*}
         \sum_{j} r' p''_j\left[\demandnb_{ij}(\bbp', e_i) - \demandnb_{ij}(\bbp'', e_i)\right] &= r'\left( \sum_{j \in \mathcal{S}} (r / r') p'_j  \demandnb_{ij}(\bbp', e_i) + \sum_{j \notin S} p'_j  \demandnb_{ij}(\bbp', e_i) - e_i\right) \\
         &= r'\left( \sum_{j \in \mathcal{S}} \left[(r - r') / r'\right] p'_j  \demandnb_{ij}(\bbp', e_i) + \sum_{j} p'_j  \demandnb_{ij}(\bbp', e_i) - e_i\right) \\
         &= (r - r') \sum_{j \in \mathcal{S}} p'_j  \demandnb_{ij}(\bbp', e_i).
    \end{align*}
    
    To establish \eqref{eqn::1}, it suffices to demonstrate:
    \begin{align*}r' \sum_{j\in \mathcal{S}} \demandnb_{ij}(\bbp'', e_i) p''_j \log (r / r') \leq (r - r') \sum_{j \in \mathcal{S}} p'_j  \demandnb_{ij}(\bbp', e_i). \numberthis \label{eqn::7}
    \end{align*}
    The inequality follows from the fact that $r' \log (r/r') \leq (r - r')$, as well as the following derivation: 
    \begin{align*}
        \sum_{j\in \mathcal{S}} \demandnb_{ij}(\bbp'', e_i) p''_j &= e_i - \sum_{j\notin \mathcal{S}} \demandnb_{ij}(\bbp'', e_i) p''_j \\
        &\leq e_i - \sum_{j\notin \mathcal{S}} \demandnb_{ij}(\bbp', e_i) p'_j  \numberthis \label{eqn::6}\\
        & =  \sum_{j\in \mathcal{S}} \demandnb_{ij}(\bbp', e_i) p'_j.
    \end{align*}
    The inequality holds due to the GS property.  
    
    When equality holds in \eqref{eqn::1}, two cases arise:  either $r = r'$, which implies $\bbp' = \bbp''$; or $\sum_{j\in \mathcal{S}} \demandnb_{ij}(\bbp'', e_i) p''_j  = \sum_{j\in \mathcal{S}} \demandnb_{ij}(\bbp', e_i) p'_j = 0$ in \eqref{eqn::7} and \eqref{eqn::6}, which implies $\demandnb_{ij}(\bbp'', e_i) = \demandnb_{ij}(\bbp', e_i)$ for $j \notin \mathcal{S}$ as $p'_j = p''_j$. In either case,  $\demand_{i}(\bbp', e_i) = \demand_{i}(\bbp'', e_i)$.

\paragraph{Proof of Lemma~\ref{lem:min-side} for Operation 2} The argument proceeds similarly to the previous case. Let
 \[ \mathcal{S} = \arg\min_j p_j / p'_j,~~~ r = \min_j p_j / p'_j, ~~~\text{and}~~~ r' = \min\left\{\min_{j \notin \mathcal{S}} p_j / p'_j~,~ 1\right\},\]
    Construct the price vector $\bbp''$ as follows:
    \begin{align*}
        p''_j = \begin{cases}
        p_j / r' & \text{if } j \in \mathcal{S}\\
        p'_j & \text{o.w. } .
    \end{cases}
    \end{align*}
    In this context, the analog of Claim~\ref{clm:spending} is
    \begin{claim}
        $\sum_{j \in \mathcal{S}} p_j \demandnb_{ij}(\bbp, e_i) \geq r' \sum_{j \in \mathcal{S}}  p''_j \demandnb_{ij}(\bbp'', e_i)$.
    \end{claim} 
    Since $r < r'$, it follows that $\log (r / r') < 0$. Again, our objective is to prove 
    \begin{align*}
        r' \sum_{j\in \mathcal{S}} \demandnb_{ij}(\bbp'', e_i) p''_j \log (r / r') \leq  \sum_{j} p_j\left[\demandnb_{ij}(\bbp', e_i) - \demandnb_{ij}(\bbp'', e_i)\right].\numberthis \label{eqn::5}
    \end{align*}
    The lower bound on the RHS still follows:
    \begin{align*}
        &\sum_{j} p_j\left[\demandnb_{ij}(\bbp'', e_i) - \demandnb_{ij}(\bbp', e_i)\right]\\
        &~~~~~~=\sum_{j \in \mathcal{S}} r' p''_j \left[\demandnb_{ij}(\bbp', e_i) - \demandnb_{ij}(\bbp'', e_i)\right] + \sum_{j \notin \mathcal{S}} p_j\left[\demandnb_{ij}(\bbp', e_i) - \demandnb_{ij}(\bbp'', e_i)\right]  \\
        &~~~~~~\geq \sum_{j \in \mathcal{S}} r' p''_j \left[\demandnb_{ij}(\bbp', e_i) - \demandnb_{ij}(\bbp'', e_i)\right] + \sum_{j \notin \mathcal{S}} r' p''_j\left[\demandnb_{ij}(\bbp', e_i) - \demandnb_{ij}(\bbp'', e_i)\right] \\
        &~~~~~~ = \sum_{j} r' p''_j(\demandnb_{ij}(\bbp', e_i) - \demandnb_{ij}(\bbp'', e_i)) \\
        &~~~~~~ = (r - r') \sum_{j \in \mathcal{S}} p'_j  \demandnb_{ij}(\bbp', e_i).
    \end{align*}
    The inequality holds because, for $j \notin \mathcal{S}$, $p_j \geq r' p''_j$ and $\demandnb_{ij}(\bbp', e_i) \geq \demandnb_{ij}(\bbp'', e_i)$, which follows from the GS property.
    Then, to prove \eqref{eqn::5}, it suffices to show:
    $$r' \sum_{j\in \mathcal{S}} \demandnb_{ij}(\bbp'', e_i) p''_j \log (r / r') \leq (r - r') \sum_{j \in \mathcal{S}} p'_j  \demandnb_{ij}(\bbp', e_i).$$ 
    The inequality is satisfied because $r' \log (r/r') \leq (r - r')$, $r- r' < 0$ and $\sum_{j\in \mathcal{S}} \demandnb_{ij}(\bbp'', e_i) p''_j  \geq  \sum_{j\in \mathcal{S}} \demandnb_{ij}(\bbp', e_i) p'_j$ which follows from  the GS property.

    This completes the proof of Lemma~\ref{lem:main-technique}.
    \end{proof}
  Furthermore,  the following corollary follows directly from our proof and will be useful in the subsequent argument. Here, we consider two prices $\bbp$ and $\bbq$, and define $\tilde{\bbq}$ to be the projection of $\bbq$ into the range $[\epsilon \bbp, \frac{1}{\epsilon} \bbp]$. Lemma~\ref{lem:main-technique} holds for both $(\bbp, \bbq)$ and $(\bbp, \tilde{\bbq})$. This corollary demonstrates that the inequality gap in Lemma~\ref{lem:main-technique} is larger for $(\bbp, \bbq)$ then for $(\bbp, \tilde{\bbq})$. This is because we can apply Operation 1 and 2 to first adjust $\bbq$ to $\tilde{\bbq}$, and then perform further adjustments.
    \begin{corollary} \label{cor:main-technique}
    For any GS utility $u_i(\cdot)$, for any price vectors $\bbp > 0$ and $\bbq > 0$, and any $1 > \epsilon > 0$, define $\tilde{\bbq}$ as follows: $$\tilde{q}_j = \begin{cases} \epsilon p_j \text{ if } q_j < \epsilon p_j \\ \frac{1}{\epsilon} p_j \text{ if } q_j > \frac{1}{\epsilon} p_j \\ q_j \text{ o.w.} \end{cases}, $$ 
    Then, we have the following inequality:
    \begin{align*} &\sum_j p_j \demandnb_{ij}(\bbp, e_i) \log \frac{p_j}{q_j}  - \sum_j p_j(\demandnb_{ij}(\bbq, e_i) - \demandnb_{ij}(\bbp, e_i))  \\
    &\hspace{1.5in}\leq \sum_j p_j \demandnb_{ij}(\bbp, e_i) \log \frac{p_j}{\tilde{q}_j}  - \sum_j p_j(\demandnb_{ij}(\tilde{\bbq}, e_i) - \demandnb_{ij}(\bbp, e_i)).
    \end{align*}
\end{corollary}
The corollary follows by adjusting from $\bbq$ to $\tilde{\bbq}$, adding terms $\sum_j p_j \demandnb_{ij}(\bbp, e_i) \log \frac{q''_j}{q'_j}  - \sum_j p_j(\demandnb_{ij}(\bbq', e_i) - \demandnb_{ij}(\bbq'', e_i)) \leq 0$ to the right-hand side. Telescoping these expressions yields the left-hand side.
\subsection{Proof of Theorem~\ref{thm:main-2}} \label{sec:proof-convergence}
In this section, we aim to demonstrate that the allocation $\bbx^t$ converges to $\bbx^*$. To achieve this, we establish several technical lemmas.

The main technical challenge is to deal with the case when $x^*_{ij} = 0$ for some buyer $i$ and good $j$, as the corresponding price may not always be well-defined in this case. To address this, we require the following technical lemmas, which are later used in the main proof. Our first lemma shows that, given an allocation, there exists a threshold such that for any price vector resulting in this allocation, the price vector must be no less than this threshold.
\begin{lemma}\label{lem:tech-2}
    Given an allocation $\bbx_i = \demand_i(\bbp, e_i)$ and $\bbp > 0$, there exists a threshold $\epsilon > 0$ such that $\bbp' \geq \epsilon$ for any $\bbp'$ satisfying $\bbx_i = \demand_i(\bbp', e_i)$.
\end{lemma}
\begin{proof}
    We construct a sequence of allocations $\{\bbx^{(k)}_i\}_{k = 1}^m$ such that $u_i(\bbx^{(k)}_i) > u_i(\bbx_i)$ for all $k$. These allocations take the form:
    \begin{align*}\bbx^{(k)}_i = \alpha \bba_k + (1 - \beta)\bbx_i, \end{align*} where $\alpha> 0$, $\beta > 0$ and $\bba_k$ is the allocation vector that assigns a unit of the $k$-th good and zero of all other goods.
     We first construct $\alpha_k$ and $\beta_k$ such that $u(\alpha_k a_k + (1 - \beta_k) x_i) > u(x_i)$. Since utility is strictly increasing, we have $u(a_k + x_i) > u(x_i)$. By reducing $x_i$ slightly, choosing a small enough $\beta_k$, we ensure $u(a_k + (1- \beta_k) x_i) > u(x_i)$. Setting $\alpha_k=1$, we have $u(\alpha_k a_k + (1- \beta_k) x_i) > u(x_i)$. Define $\alpha = \alpha_k = 1$ and $\beta = \min_k {\beta_k}$ and set $x^{(k)}_i = \alpha a_k + (1- \beta) x_i$. Since $x^{(k)}_i \ge \alpha_k a_k + (1- \beta_k) x_i$, it follows that $u(x^{(k)}_i) > u(x_i)$.

    Next, let $\epsilon = \beta e_i / (2\alpha)$. If $p'_k < \epsilon$, then: 
    \begin{align*}
        (\bbp')^\top \bbx^{(k)}_i &= \alpha (p'_k)  + (1 - \beta) (\bbp')^\top \bbx_i \\
        & \leq  \epsilon \alpha   + (1 - \beta) e_i \\
        &< e_i.
    \end{align*}
    This shows that $\bbx^{(k)}_i$ is budget feasible. However, since $u_i(\bbx^{(k)}_i) > u_i(\bbx^*_i)$ for all $k$, this contradicts the assumption that $\bbx_i = \demand_i(\bbp', e_i)$.
\end{proof}

Suppose two price vectors, $\bbp$ and $\bbp'$, result in the same allocation $\bbx_i$. It is possible, for some good $j$ with $x_{ij} = 0$, $p'_j \gg p_j$. The following lemma demonstrates that for any $\epsilon < 1$, we can construct a new price vector $\tilde{\bbp}  = \min\{\bbp/ \epsilon, \bbp'\}$, which reduces the gap between $\bbp, \bbp'$, while still resulting in the same allocation $\bbx_i$. 
\begin{lemma}\label{lem:tech-3}
    Given an allocation $\bbx_i$ and price vector $\bbp > 0$ and $\bbp' > 0$ such that $\bbx_i = \demand_i(\bbp, e_i) = \demand_i(\bbp', e_i)$. For any $\epsilon < 1$, let $\tilde{\bbp}  = \min\{\bbp/ \epsilon, \bbp'\}$. Then, $\demand_i(\tilde{\bbp}, e_i) = \bbx_i$. 
\end{lemma}
\begin{proof}
    By Lemma~\ref{lem:tech-1}, $p'_j \neq p_j$ is only possible for goods $j$ where $x_{ij} = 0$. Now, assume that $\demand_i(\tilde{\bbp}, e_i) \neq \bbx_i$. By the GS property, by adjusting price from $\bbp'$ to $\tilde{\bbp}$, we have
    \begin{enumerate}[(1)]
    \item  for good $j$ where $p_j < p_j / \epsilon = \tilde{p}_j < p'_j$, $x_{ij} = 0$, and hence $\demandnb_{ij}(\tilde{\bbp}, e_i) \geq x_{ij} = 0$;
    \item  for good $j$ where $p_j \neq p'_j \leq p_j / \epsilon$, $\demandnb_{ij}(\tilde{\bbp}, e_i) = 0 $ (by GS property);
    \item  for all other goods, $\demandnb_{ij}(\tilde{\bbp}, e_i) \leq x_{ij}$ (by GS property).
    \end{enumerate}
    Observe that $\bbx_i$ is a budget-feasible bundle under price $ \tilde{\bbp}$, since $\bbx_i^\top \tilde{\bbp} \leq \bbx_i^\top \bbp' = e_i$.
    By the definition of demand, $\demand_i(\tilde{\bbp}, e_i)$ is a strictly better bundle than $\bbx_i$. 
    On the other hand, $\demand_i(\tilde{\bbp}, e_i)$ is a budget-feasible bundle under price $\bbp$, since $\demand_i(\tilde{\bbp}, e_i)^\top \bbp \leq \demand_i(\tilde{\bbp}, e_i)^\top \tilde{\bbp} = e_i$ according to (1)-(3), which contradicts the assumption that $\demand_i(\bbp, e_i) = \bbx_i$. 
\end{proof}

The next technical lemma shows that if the allocation converges to a particular allocation, the corresponding spending also converges to the spending at that allocation.
\begin{lemma} \label{lem:converge-allocation-1}
    Let $\{\bbx_i^s > 0\}_s$ be a sequence of allocations and let $\{\bbp^s\}_s$ be the corresponding prices, $\bbx_i^s = \demand_i(\bbp^s, e_i)$. If $\bbx_i^s$ converges to $\bbx_i$ when $s \rightarrow \infty$ and $\bbx_i = \demand_i(\bbp, e_i)$, then $p_j^s x^s_{ij}$ converges to $p_j x_{ij}$. This also implies $p_j^s$ converges to $p_j$ if $x_{ij} > 0$.
\end{lemma}
\begin{proof}
    We prove this result by contradiction. Suppose that $p_j^s x^s_{ij}$ does not converge to $p_j x_{ij}$. Note that $\sum_j p_j^s x_{ij}^s = e_i$. There must exist a subsequence $s_k$ such that $\{p_j^{s_k} x_{ij}^{s_k}\}_j$ converges to some limit, $\bbb_{i}^+$, such that $\bbb_{i}^+ \neq \bbb_i $ where $b_{ij} \triangleq p_j x_{ij}$.

    We divide all goods into three sets: $\mathcal{N}^+ \triangleq \{j ~|~ b_{ij}^+ > b_{ij}\}$, $\mathcal{N}^- \triangleq \{j ~|~  b_{ij}^+ < b_{ij}\}$, and $\mathcal{N}^= \triangleq \{j ~|~ b_{ij}^+ = b_{ij}\}$. 
    Since $\bbb_{i}^+ \neq \bbb_i $, $\mathcal{N}^+ \neq \emptyset$ and $\mathcal{N}^- \neq \emptyset$. Additionally, define $\mathcal{Z} \triangleq \{j ~|~ x_{ij} = 0\}$. 

    For a small enough $\epsilon > 0$ such that $b_{ij}^+ > (1 + 2\epsilon) b_{ij}$ for all $j \in \mathcal{N}^+$, there exists a $k_0$, such that for any $k \geq k_0$, 
    \begin{itemize}
        \item $p_j^{s_k} > (1 + \epsilon) p_j$ for all $j \in \mathcal{N}^+$  (as $ b_{ij}^+ > (1 + 2\epsilon) b_{ij}$ and $x_{ij}^{s_k} \rightarrow x_{ij}$);
        \item $p_j^{s_k} > (1 + \epsilon) p_j$ implies $j \in \mathcal{N}^+ \cup \mathcal{Z}$ (as $p_j^{s_k} x_{ij}^{s_k} \rightarrow b_{ij}^+ \leq b_{ij} = p_j x_{ij}$ for $j \in \mathcal{N}^- \cup \mathcal{N}^=$ and $x_{ij}^{s_k} \rightarrow x_{ij}$);
        \item $\{j~|~j \in \mathcal{Z} \text{ and } p^{s_k}_j \leq (1 + \epsilon) p_j \} \subseteq \mathcal{Z} \cap \mathcal{N}^=$ (as $\mathcal{Z} \cap \mathcal{N}^- = \emptyset$ by definition and $p_j^{s_k} > (1 + \epsilon) p_j$ for all $j \in \mathcal{N}^+$).
    \end{itemize}
    Therefore, Given a small $\epsilon > 0$, consider projecting the prices $p_j^{s_k}$ to lie within the range $[(1 - \epsilon) p_j, (1+ \epsilon) p_j]$, i.e., define $\tilde{p}_j^{s_k} \triangleq \max\left\{(1 - \epsilon) p_j, \min\{(1+ \epsilon) p_j, p_j^{s_k}\}\right\}$.
    Thus, 
    \begin{align*}
        \sum_{j \in \mathcal{N}^+ \cup \mathcal{Z}} \tilde{p}_j^{s_k} \demandnb_{ij}(\tilde{\bbp}^{s_k}, e_i) &\geq \sum_{j \in \mathcal{N}^+ \cup \mathcal{Z} \setminus \{j|~j \in \mathcal{Z} \text{ and } p^{s_k}_j \leq (1 + \epsilon) p_j^* \}} \tilde{p}_j^{s_k} \demandnb_{ij}(\tilde{\bbp}^{s_k}, e_i) \\
        &\geq \sum_{j \in \mathcal{N}^+ \cup \mathcal{Z} \setminus \{j|~j \in \mathcal{Z} \text{ and } p^{s_k}_j \leq (1 + \epsilon) p_j^* \}} p_j^{s_k} \demandnb_{ij}(\bbp^{s_k}, e_i) \\
        &\geq \sum_{j \in \mathcal{N}^+ \cup \mathcal{Z} \setminus \mathcal{N}^=} p_j^{s_k} \demandnb_{ij}(\bbp^{s_k}, e_i) \\
        &\geq \sum_{j \in \mathcal{N}^+ } p_j^{s_k} \demandnb_{ij}(\bbp^{s_k}, e_i).
    \end{align*}
    For the second inequality, we consider the price adjustments from $\bbp^{s_k}$ to $\tilde{\bbp}^{s_k}$. We reduce the prices from $p_j^{s_k}$ to $\tilde{p}_j^{s_k}$ only for $j \in \mathcal{N}^+ \cup \mathcal{Z} \setminus \{j|~j \in \mathcal{Z} \text{ and } p^{s_k}_j \leq (1 + \epsilon) p_j^* \}$ and we may increase other goods' prices. Therefore, by the GS property, the price adjustments will increase the total spending on goods in $\mathcal{N}^+ \cup \mathcal{Z} \setminus \{j|~j \in \mathcal{Z} \text{ and } p^{s_k}_j \leq (1 + \epsilon) p_j^* \}$.
    
    Letting $\epsilon \rightarrow 0$ and $k \rightarrow \infty$ accordingly, the LHS converges to $\sum_{j \in \mathcal{N}^+ \cup \mathcal{Z}} p_j \demandnb_{ij}(\bbp, e_i) = \sum_{j \in \mathcal{N}^+ \cup \mathcal{Z}} b_{ij} = \sum_{j \in \mathcal{N}^+} b_{ij}$, while the RHS converges to $\sum_{j \in \mathcal{N}^+} b_{ij}^+$. This leads to a contradiction as $\mathcal{N}^+ = \{j ~|~ b_{ij}^+ > b_{ij}^*\} \neq \emptyset$.
\end{proof}

Now, we prove Theorem~\ref{thm:main-2}.

\begin{proof}[Proof of Theorem~\ref{thm:main-2}]
Assume for contradiction that the sequence of allocations $\bbx^t$ doesn't converge to $\bbx^*$. This implies there exists an agent $i$ such that $\bbx^t_i$ doesn't converge to $\bbx^*_i$. Since $\bbx^t_i \in [0, 1]^m$, there must exists a subsequence $\bbx^{t_k}_i$ that converges to some $\tilde{\bbx}_i \neq \bbx^*$. We analyze two cases based on the properties of $\tilde{\bbx}_i$.

\paragraph{Case 1: $\tilde{\bbx}_i > 0$} Using Lemma~\ref{lem:tech-1} and \ref{lem:main-technique}, there exists a $\delta > 0$ such that 
\begin{align*}
    \sum_j p^*_j \demandnb_{ij}(\bbp^*, e_i) \log \frac{p^*_j}{q_{ij}(\tilde{\bbx}_i)}  - \sum_j p_j(\tilde{x}_{ij} - \demandnb_{ij}(\bbp^*, e_i)) \leq -\delta, \numberthis \label{eq:main-2-1}
\end{align*}
where $\bbq_i(\tilde{\bbx}_i)$ is the corresponding price vector of $\tilde{\bbx}_i$. Thus, for sufficiently large $t$,  this implies: 
\begin{align*}\sum_{ij} b_{ij}^*  \log \frac{p_j^*}{e_i \frac{ \nabla_j u_i(\bbx_i^{t_k})}{\sum_{j'} x_{ij'}^t \nabla_{j'} u_i(\bbx_i^{t_k})}}  &\leq \sum_j p^*_j \demandnb_{ij}(\bbp^*, e_i) \log \frac{p^*_j}{q_{ij}(\bbx_i^{t_k})}  - \sum_j p^*_j(x^{t_k}_{ij} - \demandnb_{ij}(\bbp^*, e_i)) \\
&\rightarrow \sum_j p^*_j \demandnb_{ij}(\bbp^*, e_i) \log \frac{p^*_j}{q_{ij}(\tilde{\bbx}_i)}  - \sum_j p^*_j(\tilde{x}_{ij} - \demandnb_{ij}(\bbp^*, e_i)) \\
&\overset{(a)}{\leq} -\delta.
\end{align*} 
(a) is by \eqref{eq:main-2-1} and the limit holds due to Lemma~\ref{lem:converge-allocation-1}.  This leads to a contradiction because the LHS of equation \eqref{eq:3} is always positive. Thus, this case is invalid.

\paragraph{Case 2: $\tilde{\bbx}_i$ is not strictly positive.}In this case, we consider the corresponding price vectors $\bbq_i(\bbx_i^{t_k})$ and project them into the bounded domain $[\epsilon \bbp^*, \frac{1}{\epsilon} \bbp^*]$,  resulting in $\tilde{\bbq}_i^t$:
\begin{align*}
\tilde{q}_{ij}^t = \begin{cases}
    \epsilon p^*_j &\text{ if } q_{ij}(\bbx_i^t) <  \epsilon p^*_j\\
    \frac{1}{\epsilon}  p^*_j &\text{ if } q_{ij}(\bbx_i^t) >  \frac{1}{\epsilon} p^*_j\\
    q_{ij}(\bbx_i^t) &\text{ o.w.}
\end{cases}
\end{align*}  $\epsilon$ will be specified later. 
    
    By Corollary~\ref{cor:main-technique}, 
    \begin{align*}&\sum_j p^*_j \demandnb_{ij}(\bbp^*, e_i) \log \frac{p^*_j}{q_{ij}(\bbx_i^t)}  - \sum_j p^*_j\left[\demandnb_{ij}(\bbq_i(\bbx_i^t), e_i) - \demandnb_{ij}(\bbp^*, e_i)\right] \\
    &\hspace{1in}\leq \sum_j p^*_j \demandnb_{ij}(\bbp^*, e_i) \log \frac{p^*_j}{\tilde{q}^t_{ij}}  - \sum_j p^*_j\left[\demandnb_{ij}(\tilde{\bbq}^t_i, e_i) - \demandnb_{ij}(\bbp^*, e_i)\right].  \numberthis \label{eq::gs-converg-1}
    \end{align*}
    Since  $[\epsilon \bbp^*, \frac{1}{\epsilon} \bbp^*]$ is bounded, there exists a further subsequence of subsequence $\{\bbx_i^{t_k}\} \rightarrow \tilde{\bbx}_i$ such that $\tilde{\bbq}^{t_k}_i \rightarrow \tilde{\bbq}^+_i$. Note that
    
    \begin{align*}\sum_{ij} b_{ij}^*  \log \frac{p_j^*}{e_i \frac{ \nabla_j u_i(\bbx_i^{t_k})}{\sum_{j'} x_{ij'}^{t_k} \nabla_{j'} u_i(\bbx_i^{t_k})}}  &\leq \sum_j p^*_j \demandnb_{ij}(\bbp^*, e_i) \log \frac{p^*_j}{q_{ij}(\bbx_i^{t_k})}  - \sum_j p^*_j(x^{t_k}_{ij} - \demandnb_{ij}(\bbp, e_i)) \\
    & \leq \sum_j p^*_j \demandnb_{ij}(\bbp^*, e_i) \log \frac{p^*_j}{\tilde{q}^{t_k}_{ij}}  - \sum_j p^*_j(\demandnb_{ij}(\tilde{\bbq}^{t_k}_i, e_i) - \demandnb_{ij}(\bbp^*, e_i))\quad \text{(by \eqref{eq::gs-converg-1})} \\
    &\rightarrow \sum_j p^*_j \demandnb_{ij}(\bbp^*, e_i) \log \frac{p^*_j}{\tilde{q}^{+}_{ij}}  - \sum_j p^*_j(\demandnb_{ij}(\tilde{\bbq}^+_i, e_i) - \demandnb_{ij}(\bbp^*, e_i)).
    \end{align*}
    If one can show $\demand_i(\tilde{\bbq}^+_i, e_i) \neq \bbx_i^*$, the RHS is no larger than $-\delta$ for some positive $\delta$ (by Lemma~\ref{lem:main-technique}), which contradicts the positivity of the LHS of equation \eqref{eq:3}.

    Assume for a contradiction that $\demand_i(\tilde{\bbq}^+_i, e_i) = \bbx_i^*$ but $\tilde{\bbx}_i \neq \bbx_i^*$ for small enough $\epsilon$. We choose $\epsilon$ small enough such that $\tilde{\bbq}^+_i > 2 \epsilon \bbp^*$ (by Lemma~\ref{lem:tech-2}). 
    In this case, for sufficiently large $t_k$, 
    \begin{enumerate}
        \item $q_{ij}(\bbx_i^{t_k}) = \tilde{q}^{t_k}_{ij} \rightarrow \tilde{q}^+_{ij} = p^*_j$ for all $j$ with $x_{ij}^* > 0$ (by Lemma~\ref{lem:tech-1});
        \item $q_{ij}(\bbx_i^{t_k}) \geq \tilde{q}^{t_k}_{ij}  \rightarrow \tilde{q}^+_{ij}$  for all $j$ (as $\tilde{\bbq}^+_i > 2 \epsilon \bbp^*$ by Lemma~\ref{lem:tech-2}).
    \end{enumerate}
    
    Thus, for sufficiently large $t_k$, for any $j$ with $x_{ij}^* > 0$, 
    $$ x_{ij}^{t_k} = \demandnb_{ij}(\bbq_i(\bbx_i^{t_k}), e_i) \geq \demandnb_{ij}(\tilde{\bbq}^{t_k}_i, e_i) \rightarrow \demandnb_{ij}(\tilde{\bbq}^{+}_i, e_i) = \bbx_{ij}^*.$$
    The inequality holds by GS property as $q_{ij}(\bbx_i^{t_k}) = \tilde{q}^{t_k}_{ij}$ for $j$ with $x_{ij}^* > 0$ and $q_{ij}(\bbx_i^{t_k}) \geq \tilde{q}^{t_k}_{ij}$ for all $j$. 
    
    Therefore,   \begin{align*}e_i &\geq \sum_{j: x_{ij}^* > 0} q_{ij}(\bbx_i^{t_k}) \demandnb_{ij}(\bbq_i(\bbx_i^{t_k}), e_i) \geq  \sum_{j: x_{ij}^* > 0} q_{ij}(\bbx_i^{t_k}) \demandnb_{ij}(\tilde{\bbq}^{t_k}_i, e_i) \\
    & \quad \quad \quad \quad= \sum_{j: x_{ij}^* > 0} \tilde{\bbq}^{t_k}_{ij} \demandnb_{ij}(\tilde{\bbq}^{t_k}_i, e_i) \overset{(a)}{\rightarrow} \sum_{j: x_{ij}^* > 0} p_j^* \bbx_{ij}^* = e_i. \numberthis \label{ineq::9}\end{align*}
    (a) is by Lemma~\ref{lem:converge-allocation-1}. This implies $\sum_{j: x_{ij}^* = 0} q_{ij}(\bbx_i^{t_k}) \demandnb_{ij}(\bbq_i(\bbx_i^{t_k}), e_i) \rightarrow 0$, which yields $\demandnb_{ij}(\bbq_i(\bbx_i^{t_k}), e_i) \rightarrow 0$ when $x_{ij}^* = 0$ as $q_{ij}(\bbx_i^{t_k}) > \epsilon p_j^*$ for sufficiently large $t_k$ and $p_j^* > 0$ by Observation~\ref{obs:pos-price}. Additionally, in \eqref{ineq::9}, as $\bbq_i(\bbx_i^{t_k}) \rightarrow p_j^*$ for $x_{ij}^* > 0$, $x_{ij}^{t_k} = \demandnb_{ij}(\bbq_i(\bbx_i^{t_k}), e_i) \rightarrow \bbx_{ij}^*$. This contradicts the assumption that $\tilde{\bbx}_i \neq \bbx_i^*$.
\end{proof}
\section{Setting of Exchange markets} \label{sec:exchange-market}
Consider an exchange market with $n$ agents, where each agent $i$ is initially endowed with goods, $\mathcal{G}_i$. Each agent has a utility function defined over all goods, and these utility functions satisfy Assumption~\ref{aspt:u-std} and Assumption~\ref{def:AGS}, as in Fisher markets.

In the exchange market, agents begin by selling their endowed goods to obtain money. They then use the money to purchase goods that maximize their utilities. Thus, given a price vector $\bbp$ over goods, the demand of agent $i$ is defined as:
$$ \demand_i(\bbp, e_i) =  \arg \max_{\bbx_i} \text{~s.t.~} \bbp^\top \bbx_i \leq e_i \text{~and~} \bbx_i \geq 0,$$
where $e_i = \sum_{j \in \mathcal{G}_i} p_j$.

\paragraph{Exchange market equilibrium}
An exchange market equilibrium characterizes the equilibrium allocation and price in the market. 
\begin{definition}[Exchange market equilibrium]
    An allocation-price pair $(\bbx^*, \bbp^*)$ is an exchange market equilibrium if the following conditions hold:
    \begin{enumerate}
        \item each agent $i$ receives the optimal bundle given the price, $\bbx^*_i \in \demand_i(\bbp^*, e^*_i)$, where $e^*_i = \sum_{j \in \mathcal{G}_i} p^*_j $;
        \item no good is oversold, $\sum_{i} x^*_{ij} \leq 1$ for all $j$;
        \item every good with a positive price is fully allocated, $\sum_{i} x^*_{ij} < 1$ implies $p^*_j = 0$.
    \end{enumerate}
\end{definition}
Since utility functions are strictly increasing, the following observation still holds.
\begin{observation} \label{obs:pos-price-exchange}
    At equilibrium, $\bbp^* > 0$ and for each good $j$, $\sum_i x^*_{ij} = 1$. 
\end{observation}

\subsection{Lazy Proportional Response Dynamics} 
Lazy Proportional Response Dynamics is an iterative procedure introduced by \citet{BranzeiDevanurRabani2021} in exchange markets. Each agent has a parameter $0 < \alpha_i < 1$, which specifies the proportion of her current budget agent $i$ wants to spend in the market. At each iteration $t$, the dynamics are characterized by the following steps: 
\begin{enumerate}
    \item Budget Spending: Each agent $i$ allocates a fraction, $\alpha_i$, of the total budget $B_i^t$ among the goods. Let $b^t_{ij}$ denotes the spending of agent $i$ on good $j$ at round $t$. It follows that $\sum_j b_{ij}^t = \alpha_i B_i^t$ for each agent $i$.
    \item Goods Allocation: Based on the spending $\{b^t_{ij}\}_{\{i, j\}}$, the amount of good $j$ allocated to agent $i$ is proportional to her spending: $x_{ij}^t = b^t_{ij} / (\sum_i b_{ij}^t)$.
    \item Budget update: For the next round, each agent $i$ collects all agents' spending on goods in $i$ and adds it to her budget: the budget of agent $i$, $B_i^{t+1}$, is updated to be $(1 - \alpha_j) B_i^t + \sum_{i'} \sum_{j \in \mathcal{G}_i} b_{i'j}^t$.
    \item Spending Update:  the spending $b_{ij}^{t+1}$ is updated to be proportional to $x_{ij}^t \nabla_j u_i(\bbx_i^t)$.
\end{enumerate}

The process begins with an initial budget $B_i^0 > 0$ and initial spending $b_{ij}^0 > 0$ for all agents $i$ and goods $j$.

Mathematically, letting $e_i^t \triangleq \alpha_i B_i^t$, which is the total spending of agent $i$ at round $t$, the entire procedure can be expressed as:
\begin{align*}
    &b^{t+1}_{ij} =  e_i^{t+1} \frac{x_{ij}^t \nabla_j u_i(\bbx_i^t)}{\sum_{j'} x_{ij'}^t \nabla_{j'} u_i(\bbx_i^t)}~~~\text{for all $i$ and $j$}\\
    &~~~~~~\hspace{0.4in}\text{where $x_{ij}^t = \frac{b_{ij}^t}{p_j^t}$, $p_j^t = \sum_i b_{ij}^t$, and $e_i^{t+1} = (1 - \alpha_i) e_i^t + \alpha_i \sum_{j \in \mathcal{G}_i} p_j^t$.} \numberthis \label{eq:pr-exchange}
\end{align*}

WLOG, we make the following assumption.
\begin{assumption}
    The initial total budgets equals $1$: $\sum_i B_i^0 = 1$.
\end{assumption}
Since the money will not be consumed, we have the following straightforward claim.
\begin{claim}
    At any round $t$, $\sum_i B_i^t = 1$.
\end{claim}

\subsection{Convergence Result}
The following theorem shows the sequence of allocations will converge to the exchange market equilibrium.
\begin{theorem}
    If a market equilibrium exists, then the sequence of allocations $\bbx^t$ converges to an equilibrium allocation $\bbx^*$.
\end{theorem}
\begin{proof}
    Consider a market equilibrium $(\tilde{\bbx}^*, \tilde{\bbp}^*)$. We define other parameters as follows: $\tilde{B}_j^* = \frac{1}{\alpha_j} \tilde{p}^*_j$, $\tilde{e}_i^* = \tilde{p}^*_i$, and $\tilde{b}^*_{ij} = \tilde{x}^*_{ij} \tilde{p}^*_j$. To normalize the total money to $1$, we apply the following transformations: $\bbp^* = \tilde{\bbp}^* / (\sum_i \tilde{B}_i^*)$, $B_i^* = \tilde{B}_i^* / (\sum_i \tilde{B}_i^*)$, $e_i^* = \tilde{e}^*_i / (\sum_i \tilde{B}_i^*)$, and $\bbx^* = \tilde{\bbx}^*$. Thus,
    \begin{align*}
        &\sum_{ij} b_{ij}^* \log \frac{b_{ij}^*}{b_{ij}^{t+1}} +  \sum_i \left( \frac{1 - \alpha_i}{\alpha_i}\right) e_i^* \log \frac{e_i^*}{e_i^{t+1}}  \\
        &~~~~~~~~~\hspace{0.2in}=  \sum_{ij} \frac{1}{\alpha_i} b_{ij}^* \log \frac{e_i^*}{e_i^{t+1}} + \sum_{ij} b_{ij}^* \log \frac{b_{ij}^*}{b_{ij}^t} - \sum_{ij} b_{ij}^* \log \frac{p_j^*}{p_j^t} + \sum_{ij} b_{ij}^* \log \frac{p_j^*}{e_i^* \frac{ \nabla_j u_i(\bbx_i^t)}{\sum_{j'} x_{ij'}^t \nabla_{j'} u_i(\bbx_i^t)} }   \\
        &~~~~~~~~~\hspace{0.2in}=   \sum_{ij} b_{ij}^* \log \frac{b_{ij}^*}{b_{ij}^t}  +  \sum_i \left( \frac{1 - \alpha_i}{\alpha_i}\right) e_i^* \log \frac{e_i^*}{e_i^{t}} + \sum_{ij} b_{ij}^* \log \frac{p_j^*}{e_i^* \frac{ \nabla_j u_i(\bbx_i^t)}{\sum_{j'} x_{ij'}^t \nabla_{j'} u_i(\bbx_i^t)} } \\
        &~~~~~~~~~~~~~~~\hspace{0.2in}\hspace{0.2in}+  \sum_{ij} \frac{1}{\alpha_i} b_{ij}^* \log \frac{e_i^*}{e_i^{t+1}}  - \sum_i  \left( \frac{1 - \alpha_i}{\alpha_i}\right) e_i^* \log \frac{e_i^*}{e_i^{t}} - \sum_{ij} b_{ij}^* \log \frac{p_j^*}{p_j^t}.
    \end{align*}
    Note that 
    \begin{align*}
        \sum_{ij} b_{ij}^* \log \frac{b_{ij}^*}{b_{ij}^{t+1}} + \sum_i \left( \frac{1 - \alpha_i}{\alpha_i}\right) e_i^* \log \frac{e_i^*}{e_i^{t+1}} &= \sum_{ij} b_{ij}^* \log \frac{b_{ij}^*}{b_{ij}^{t+1}} + \sum_i \left( \frac{1 - \alpha_i}{\alpha_i}\right) e_i^* \log \frac{\left( \frac{1 - \alpha_i}{\alpha_i}\right) e_i^*}{\left( \frac{1 - \alpha_i}{\alpha_i}\right)e_i^{t+1}} \\
        &= \sum_{ij} b_{ij}^* \log \frac{b_{ij}^*}{b_{ij}^{t+1}} + \sum_i  (1 - \alpha_i) B_i^* \log \frac{ (1 - \alpha_i) B_i^*}{ (1 - \alpha_i) B_i^{t+1}}, \\
    \end{align*}
    which is a KL divergence  as $(1 - \alpha_i) B_i^{t+1}$ can be viewed as the money kept by agent $i$.
    
    Thus, letting $\mathbf{KL}(t) \triangleq \sum_{ij} b_{ij}^* \log \frac{b_{ij}^*}{b_{ij}^t}  +  \sum_i \left( \frac{1 - \alpha_i}{\alpha_i}\right) e_i^* \log \frac{e_i^*}{e_i^{t}} \geq 0$, we have
    \begin{align*}
        \mathbf{KL}(t + 1) &= \mathbf{KL}(t) + \sum_{ij} b_{ij}^* \log \frac{p_j^*}{e_i^* \frac{ \nabla_j u_i(\bbx_i^t)}{\sum_{j'} x_{ij'}^t \nabla_{j'} u_i(\bbx_i^t)} } \\
        &~~~~~~~~~~~~~\hspace{0.2in}\hspace{0.2in}+  \sum_{ij} \frac{1}{\alpha_i} b_{ij}^* \log \frac{e_i^*}{e_i^{t+1}}  - \sum_i  \left( \frac{1 - \alpha_i}{\alpha_i}\right) e_i^* \log \frac{e_i^*}{e_i^{t}} - \sum_{ij} b_{ij}^* \log \frac{p_j^*}{p_j^t} \\
        &= \mathbf{KL}(t) + \sum_{ij} b_{ij}^* \log \frac{p_j^*}{e_i^* \frac{ \nabla_j u_i(\bbx_i^t)}{\sum_{j'} x_{ij'}^t \nabla_{j'} u_i(\bbx_i^t)} } \\
        &~~~~~~~~~~~~~\hspace{0.2in}\hspace{0.2in}+  \sum_{i} \frac{1}{\alpha_i} e_i^* \log \frac{e_i^*}{e_i^{t+1}}  - \sum_i  \left( \frac{1 - \alpha_i}{\alpha_i}\right) e_i^* \log \frac{e_i^*}{e_i^{t}} - \sum_{j} p_{j}^* \log \frac{p_j^*}{p_j^t}. 
    \end{align*}

    Note that $\sum_{ij}^* b_{ij}^* \log \frac{p_j^*}{e_i^* \frac{ \nabla_j u_i(\bbx_i^t)}{\sum_{j'} x_{ij'}^t \nabla_{j'} u_i(\bbx_i^t)} } \leq 0$ by Lemma~\ref{lem:main-tech}, and $$\frac{1}{\alpha_i} e_i^* \log \frac{e_i^*}{e_i^{t+1}} \leq   \sum_i  \left( \frac{1 - \alpha_i}{\alpha_i}\right) e_i^* \log \frac{e_i^*}{e_i^{t}} + \sum_{j \in \mathcal{G}_i} p_j^* \log \frac{p_j^*}{p_j^t}$$ as $e_i^{t+1} = (1 - \alpha_i) e_i^t + \alpha_i \sum_{j \in \mathcal{G}_i} p_j^t$. This implies $\mathbf{KL}(t + 1)$ is non-increasing over $t$.

    Additionally, the result follows by applying the argument in Section~\ref{sec:proof-convergence}.
 \end{proof}

\section{Conclusion}
In this paper, we introduce a new PR update rule. We demonstrate that in Fisher markets, when all buyers have GS utility functions, the new PR rule converges point-wise to the competitive equilibrium. Moreover, we establish that the PR process guarantees an empirical convergence rate of $\calO(1/T)$ towards the equilibrium price. In Arrow-Debreu markets, we focus on a lazy version of the PR update and show that it also converges point-wise to the competitive equilibrium.

One interesting direction for future research is to explore the cases where linear convergence rate can be established for PR dynamics. Previous work \cite{Zhang2011,cheung2018dynamics} analyzed the PR process with CES utility functions, showing that the convergence rate is linear when the CES parameter $\rho$ is bounded away from 1, with $\rho=1$ corresponding to the linear utility case. The question remains whether this linear convergence result can be extended to broader classes of utility functions beyond CES utilities.

Convergence results for natural market dynamics have largely been established for subfamilies of GS utilities, with only a few exceptions~\cite{cheung2018dynamics,CCD2013,CCR2012}.
Our findings contribute to this line of research, moving toward the ultimate goal of delineating the boundaries of markets in which PR and t\^atonnement remain effective.

\section*{Acknowledgement}

The work is supported by the Australian Research Council (ARC) Discovery Project under Grant No.: DP240100506.

\bibliographystyle{plainnat}
\bibliography{bib}
\end{document}